%% file: main.tex
\documentclass[conference, final]{IEEEtran}
\usepackage[nolist]{acronym}
\usepackage{tikz}
\usetikzlibrary{decorations.pathreplacing,calligraphy,dsp}
\usepackage{pgfplots}
\usepgfplotslibrary{groupplots}
\usepackage{graphicx}
\usepackage{subfloat}
\usepackage[font=footnotesize]{caption}
\usepackage{float}
\usepackage{array}
\pgfplotsset{compat=newest}
\usepackage{amsmath, amsbsy, amssymb}
\usepackage{tikzscale}
\usepackage{color}
\usepackage{epigraph}
\usepackage{circuitikz}
\usepackage{multirow}
\usepackage{booktabs}
\usepackage{adjustbox}
\usepackage[group-separator={,}]{siunitx}
\definecolor{darkgreen}{rgb}{0.125,0.5,0.169}
\usetikzlibrary{shapes,arrows,fadings,chains}
\usepackage{marvosym}
\usepackage{bbm}
\usepackage{mathtools}
\usepackage{amssymb}
\usepackage{amsmath}
\usepackage{tikz}
\usepackage{pgfplots}
\usepackage[inline]{enumitem}
\usepackage{amsthm}
\usepackage{xcolor}
\usepackage{bbm}
\newtheorem{lemma}{Lemma}
\newtheorem{corollary}{Corollary}
\newtheorem{proposition}{Proposition}

\usetikzlibrary{positioning,fit,calc}
\usetikzlibrary{calc,patterns,angles,quotes}
\usetikzlibrary{shapes.geometric, arrows}
\usetikzlibrary{fit}
\usepgfplotslibrary{colormaps}
\usetikzlibrary{shapes.multipart}
\usetikzlibrary{positioning,fit,calc}
\usetikzlibrary{calc, positioning}
\usetikzlibrary{decorations}
\usetikzlibrary{decorations.pathreplacing}
\usetikzlibrary{arrows.meta,shapes.arrows}
\usetikzlibrary{shapes}
\usetikzlibrary{positioning,fit,backgrounds}
\usetikzlibrary{matrix}
\usetikzlibrary{arrows.meta,decorations.pathreplacing}
\usetikzlibrary{
    arrows.meta,
    positioning,
    decorations.pathreplacing,
    calc
}
\usepgfplotslibrary{groupplots}
\usepackage{algorithm}
\usepackage{algorithmic}
\usetikzlibrary{arrows.meta,calc,positioning,fit,backgrounds}
\pgfplotsset{compat=1.13}
\tikzset{>=latex}
\usepackage{tikzit}
\input{macros.tex}

\input{acronyms.tex}
\input{corporateColours.tex}

\input{plotcyclelists.tex}
\input{system.tikzstyles}
\usepackage{pgfplots}
\usepackage{pgfplotstable}
\usepackage{tikz}
\usepackage{braket}
\usepackage{comment}
\usepackage{dblfloatfix}
\usepackage{adjustbox}
\usepackage[caption=false,font=footnotesize]{subfig}

\usetikzlibrary{arrows.meta,positioning,calc}

\pgfplotsset{
    discard if not/.style 2 args={
        x filter/.code={
            \edef\tempa{\thisrow{#1}}
            \edef\tempb{#2}
            \ifx\tempa\tempb
            \else
                
            \fi
        }
    }
}
\pgfplotsset{
    colormap={mysummer}{
        color(0cm)=(green!60!black)
        color(1cm)=(yellow!85!white)
    }
}
\pgfplotsset{compat=1.18}

\definecolor{violett}{RGB}{128,0,128}
\IEEEoverridecommandlockouts
\begin{document}

\title{Auxiliary Nodes for BP Decoding of Quantum LDPC Codes}

\author{\IEEEauthorblockN{Daniel Tandler, Paul Bezner, and Stephan ten Brink}

\IEEEauthorblockA{
Institute of Telecommunications, University of Stuttgart, Pfaffenwaldring 47, 70569 Stuttgart, Germany \\
\{tandler, bezner, tenbrink\}@inue.uni-stuttgart.de
}

}

\maketitle
\begin{abstract}
Many recently proposed \ac{CSS} \ac{QLDPC} codes have sparse decoding
graphs, enabling syndrome-based \ac{BP} decoding at low complexity.
Their construction, however, often results in properties that impair \ac{BP}
performance, such as short cycles and degeneracy. 
In this work, we propose a general framework for introducing \acp{AVN} and \acp{ACN} into the decoding graph of \ac{CSS} codes, compatible with the standard stabilizer measurement framework.
This provides an additional degree of freedom in the design of the decoding graph itself and can be used to tackle the aforementioned shortcomings.
We show that recently proposed techniques, $4$-cycle removal and subcode ensemble decoding, can be interpreted as instances of this framework.
For $4$-cycle removal, we find that the gains depend strongly on the \ac{BP} iteration count and check-node message scaling. 
Building on this framework, we further propose a graph-derived subcode ensemble decoder and demonstrate under circuit-level noise that it substantially reduces the per-round logical error rate compared with \ac{BP} on the corresponding $4$-cycle-free decoding graph.
\end{abstract}

\acresetall

\section{Introduction}
\Ac{BP} decoding is attractive for \ac{QLDPC} codes due to its low complexity and direct compatibility with sparse parity-check representations. However, the performance of \ac{BP} decoding is highly sensitive to the chosen Tanner-graph representation.
In particular, trapping sets, often consisting of \emph{short cycles}, and \emph{degeneracy}-induced symmetries can strongly degrade the decoding performance.
Short cycles in the Tanner graph cause the messages in \ac{BP} to become strongly correlated and thus oscillate and reinforce potential wrong beliefs. 
Degeneracy refers to the phenomenon that can occur in syndrome based decoding, where multiple error hypotheses can share the same syndrome.
While in theory finding a single one of these hypotheses can be sufficient in the quantum setting, this ambiguity can strongly impair the \ac{BP} decoder's performance.
This is especially strong in \ac{QLDPC} codes as they tend to have low-weight stabilizers.
One line of research aims to circumvent these issues by applying different decoding methods, usually after a first pass of \ac{BP} decoding \cite{decision_tree,ambiguity,roffe_decoding_2020}.
While improving performance, this can also increase the worst-case latency, which might be problematic for practical \ac{QEC}.
Other approaches \cite{müller2025improvedbeliefpropagationsufficient, koutsioumpas2025colourcodesreachsurface, CAGNN, ye2025beamsearchdecoderquantum}, aim to introduce decoding diversity by changing the decoding trajectory that \ac{BP} sees during different decoding runs.
Yet other approaches try to address these issues more directly by modifying the concrete decoding graph that \ac{BP} uses:
In \cite{KIT1}, a neural decoder for quaternary \ac{BP} together with \emph{overcomplete} parity check matrices has been investigated. 
\cite{KIT2} proposes a joint code and decoder scheme by concentrating all $4$-cycles in the parity check matrix of the code on a single \ac{VN}, effectively removing them from the graph when using a decoder ensemble on that \ac{VN}.
\cite{yin2024symbreak} proposes to \emph{break} potentially problematic checks that cause degeneracy issues after a first \ac{BP} decoding pass.
Structural decoding graph modifications in the form of $4$-cycle removal for the decoding graph of \ac{BP} decoders using \ac{ACN} and \ac{AVN} were first introduced for classical channel codes \cite{yedidia2002generating, yifei_1}.
Recently, \cite{gari} introduced a related graph augmentation (\ac{GARI}) specifically targeting the $4$-cycles induced by $Y$-type errors in the \ac{DEM}, while proving equivalence of the decoding problem for their construction.
\Ac{ASCED} was first introduced for classical codes in \cite{mandelbaum2026affinesubcodeensembledecoding}.
In \cite{wursthorn2026affine}, \ac{ASCED} is adapted for the Quantum Decoding setting by introducing additional linear independent parity check constraints acting as \emph{degeneracy splitters}.
In this work, we show that these two latter viewpoints, short cycle removal and subcode ensemble approaches, can be connected through a common framework.
\begin{figure}[!t]
    \centering
        \begin{adjustbox}{width=0.5\textwidth}
    \input{fig/system_model_updated.tex}
      \end{adjustbox}
    \caption{Decoding architecture of the considered framework: $r$ syndrome measurement rounds of an encoded quantum state yield the syndrome string $\mathbf{s}_{\mathrm{DEM}}$. Using the linear map $\mathbf{A}$, this syndrome is mapped to $\mathbf{s}_{\mathrm{aug}}$, which is suitable for decoding on the augmented parity-check matrix $\mathbf{H}_{\mathrm{aug}}$. After syndrome-based decoding, the proposed correction $\hat{\mathbf{e}}_{\mathrm{DEM}}$ is output.}
    \label{fig:system_model}
\end{figure}
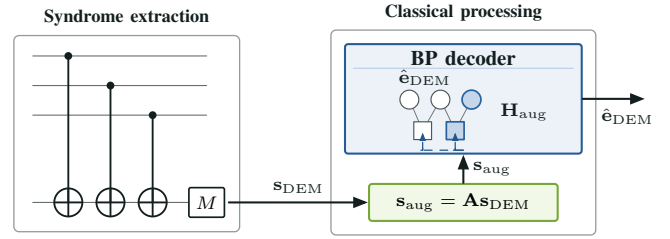
We make the following contributions:
\begin{itemize}
    \item We formalize a general framework for augmenting the decoding graph of \ac{CSS} codes with \acp{AVN} and \acp{ACN}, fully compatible with standard stabilizer measurements.\footnote{No additional stabilizer or detector measurements are required.}

    \item We show this construction has two distinct interpretations: \acp{ACN} can be used for structural modifications of the decoding graph, e.g. for removing $4$-cycles, while fixing an \ac{AVN} imposes a potentially degeneracy-splitting affine subcode constraint.

    \item We construct a graph-derived affine subcode ensemble by reusing the \acp{ACN} and \acp{AVN} from $4$-cycle removal as candidate constraints, in both a fixed and an adaptive (\ac{BP}-guided) variant.
\end{itemize}
The complete decoding architecture for the proposed method is shown in Fig. \ref{fig:system_model}.
\textbf{\textit{Notation}}: 
Scalars are denoted by lowercase, non-bold letters; vectors by lowercase, bold letters; and matrices by uppercase, bold letters. Unless stated otherwise, all vectors are column vectors.
Furthermore, $\mathrm{rank}(\mathbf{H})$ denotes the rank of $\mathbf{H}$ over $\mathbb{F}_2$, $\mathrm{row}(\mathbf{H})$ its row space, $\mathrm{supp}(\mathbf{h}^T)$ the support of row-vector $\mathbf{h}^T$ and $\mathbb{S}$ denotes a linear space.
If not mentioned otherwise, operations are assumed to be performed in the binary field $GF(2)$.
For $\mathbf{H}\in\mathbb{F}_2^{m\times n}$ and $\mathbf{s}\in\mathbb{F}_2^{m}$ we write
$\mathcal{E}(\mathbf{H},\mathbf{s})\coloneqq\{\mathbf{x}\in\mathbb{F}_2^{n}:\mathbf{H}\mathbf{x}=\mathbf{s}\}$,
and $\pi$ for the projection of $\mathbb{F}_2^{n'}$,$n'>n$ onto its first $n$ coordinates.
We denote the $i$th row of a matrix $\mathbf{H}$ by $\mathbf{h}\tp _i$. For an index set $\mathcal{S}\subseteq{1,\ldots,n}$, $\mathbbm{1}_{\mathcal{S}}\in\mathbb{F}_2^n$ denotes its indicator vector, whose $j$th entry equals $1$ if $j\in\mathcal{S}$ and $0$ otherwise. The support of a binary vector $\mathbf{x}$ is denoted by $\operatorname{supp}(\mathbf{x})$.
\section{Preliminaries}
\label{sec:prelim}
\subsection{System Model} \label{sec:system_model_section}
In this work, we consider \ac{QEC} using \ac{CSS} stabilizer codes. 
A \ac{CSS} code can be described by a pair of binary matrices $\mathbf{H}_X$, corresponding to $X$-type stabilizers, and $\mathbf{H}_Z$, corresponding to $Z$-type stabilizers.
Error detection and correction can be performed by \emph{measuring} the stabilizers and using the binary measurement results to infer the error location.
In practice this syndrome extraction is realized using a sequence of quantum gates, resulting in a measurement \emph{circuit}.
As this process can be noisy itself, it can introduce additional errors.
This work adopts the widely recognized \emph{circuit-level noise model} (together with the corresponding \ac{DEM}) for Quantum Error Correction as described in e.g. \cite{Derks_2025, Bravyi_2024, Kang2025quitsmodularqldpc}, which can also model these errors. 
Here, it is assumed that independent errors can occur at any position in the syndrome extraction circuit (state preparation, measurement, single- and two-qubit gates).
We assume an equal probability $p$ for each error class.
To account for these additional potential errors, $r$ repeated syndrome measurement rounds are performed.\footnote{Following conventions, we assume that the final measurement round is noiseless. We furthermore assume for simplicity and ease of exposition a $Z$-logical memory experiment where only the $Z$-type stabilizers are measured.}
This measurement circuit produces a collection of binary measurement outcomes from which one can form binary parities whose values are deterministic in the absence of errors.
These parities are usually referred to as \emph{detectors}.
Logical observables are similar to detectors with the difference that they only have a known value if the logical state is known.
Altogether, the decoder-relevant quantities in this model can be arranged in the following way \cite{Derks_2025} \cite{ambiguity}:
\begin{itemize}
\item Detector matrix $\mathbf{H}_{\text{DEM}} \in \mathbb{F}_2^{M \times N}$: Each row corresponds to a detector and each column to an error location. $\mathbf{H}_{ij}=1$ if the $j$-th error flips the $i$-th detector
\item Logical matrix $\mathbf{L} \in \mathbb{F}_2^{K \times N}$: Each row corresponds to a logical observable and each column to an error location. $\mathbf{L}_{ij}=1$ if the $j$-th error flips the $i$-th logical observable.
\item Prior vector $\tau \in \mathbb{R}^N$: Probability that each error occurs. 
\end{itemize}
As some errors can have an identical (no) effect on the measurements and logical observables (duplicated (all zero) columns in $\mathbf{H}_{\text{DEM}}$ and $\mathbf{L}$ ), the duplicate (all zero) columns can be deleted and their probabilities combined (removed).
Furthermore, for each instance of the noisy syndrome extraction circuit, the following can be defined: 
\begin{itemize}
\item Error vector $\mathbf{e}_{\text{DEM}} \in \mathbb{F}_2^N$: Each entry corresponds to an individual error that occurred in the syndrome extraction circuit.
\item Measured syndrome $\mathbf{s}_{\text{DEM}} \in \mathbb{F}_2^M$: Each entry corresponds to the measurement outcome of the corresponding detector.
\item Logical flips $\mathbf{l} \in \mathbb{F}_2^K$: $l_k$ indicates if $\mathbf{e}$ flipped the $k$-th logical observable.
\end{itemize}
Using this formalism, the decoding problem can be formulated as:
\begin{equation} \label{eq:system_model}
    \mathbf{s}_{\text{DEM}} = \mathbf{H}_{\text{DEM}}\mathbf{e}_{\text{DEM}}
\end{equation}
The task of the decoder is now to find an $\mathbf{\hat{e}} \in \mathbb{F}_2^N$ such that $ \mathbf{s}_{\text{DEM}}  = \mathbf{\hat{s}} = \mathbf{H}_{\text{DEM}}\mathbf{\hat{e}}$ and 
\begin{equation} \label{eq:system_logical}
    \mathbf{L(\mathbf{e}_{\text{DEM}} +\mathbf{\hat{e}} )} = \mathbf{0}
\end{equation}
thus, no logical error remains after applying the correction $\mathbf{\hat{e}}$.
\subsection{Syndrome-Based Belief Propagation (BP) Decoding}
Similar to their classical counterparts, \ac{QLDPC} codes can be decoded using a variant of the message passing \ac{BP} decoder operating on the decoding graph of the code, derived from the \ac{DEM}.
The decoder iteratively exchanges scalar messages between the graph's \acp{VN} and \acp{CN} and outputs an estimate of the marginal probability of error for each \ac{VN}.
Decoding is finished if a maximum number of iterations $I$ is reached.
The initial information for the decoder is given by the measured syndrome vector $\mathbf{s}$ and the initial channel statistics. 
The variable nodes are first initialized to the channel statistics  $L_{\mathrm{ch},j}$. 
The outgoing message $L_{c_i \rightarrow v_j}$ from \ac{CN} $c_i$ to connected \ac{VN} $v_j$ is given by
\begin{equation}\label{eq:cn_update}
  L_{c_i \rightarrow v_j} = \alpha (-1)^{s_i}2\cdot \tanh^{-1}\left( \prod_{j' \in \mathcal{N}(c_i)\setminus\{j\}} \tanh \left(\frac{ L_{v_{j'}\rightarrow c_i}}{2} \right) \right)
\end{equation}
where $\mathcal{N}(\cdot)$ denotes the set of connected nodes, $L_{v_{j'}\rightarrow c_i}$ the incoming message from the \ac{VN} $v_{j'}$ to \ac{CN} $c_i$, $s_i$ the syndrome bit at position $i$ and $\alpha$ a scalar scaling factor. 
The outgoing message $ L_{v_j \rightarrow c_i} $ from variable node $v_j$ to \ac{CN} $c_i$ is determined using
 \begin{equation}\label{eq:vn_update}
  L_{v_j \rightarrow c_i} = L_{\mathrm{ch},j} + \sum_{i' \in \mathcal{N}(v_j)\setminus\{i\}}L_{c_{i'}\rightarrow v_j}.
\end{equation}
The \emph{a posteriori} \ac{LLR} for $v_j$ is the sum of all incoming messages plus its initial value:
\begin{equation} \label{eq:final}
   L_{v_j} = L_{\mathrm{ch},j} + \sum_{i \in \mathcal{N}(v_j)}L_{c_{i}\rightarrow v_j}.
\end{equation}
Finally, the estimate $\hat{e}_j$ for the $j$-th fault location is computed by taking a hard decision on $ L_{v_j}$.
Syndrome-based \ac{BP} decoders were shown to share failure modes with 
\ac{LLR}-based \ac{BP} decoders, denoted as \emph{classical-type} trapping 
sets in \cite{raveendran2021trapping}. 
Often, these trapping sets consist 
of cycles in the decoding graph, which cause the messages exchanged during 
decoding to become highly correlated and frequently result in decoding 
failure. 
In particular, $4$-cycles (arising when two rows of the code's 
parity-check matrix overlap in two or more positions) are among the most 
harmful short cycles. 
Furthermore, syndrome-based decoders also suffer 
from degeneracy-induced symmetries (where $\mathbf{e}_1 \neq \mathbf{e}_2$ 
but $\mathbf{H}\mathbf{e}_1 = \mathbf{H}\mathbf{e}_2 = \mathbf{s}$), which 
often lead to convergence failures (\emph{quantum-type} trapping sets).
\section{Proposed Method}

In this section, we propose a method to apply the technique of using
\ac{AVN} and \ac{ACN} to the syndrome-based \ac{BP} decoding problem.

Consider a parity check matrix $\mathbf{H} \in \mathbb{F}_2^{m \times n}$ of a
binary linear code with $\mathrm{rank}(\mathbf{H})=m$\footnote{In practice,
the matrices of interest such as $\mathbf{H}_{\mathrm{DEM}}$ typically
contain redundant rows, i.e., their rank is smaller than their number of
rows. All statements and proofs in this section carry over with minor modifications; we assume full row rank only to simplify the exposition.} together with a given syndrome vector
$\mathbf{s} \in \mathbb{F}_2^{m}$.
Let
\begin{equation}\label{eq:system1}
 \mathcal{E}(\mathbf{H},\mathbf{s})\coloneqq\{\mathbf{e}\in\mathbb{F}_2^n:\mathbf{H}\mathbf{e}=\mathbf{s}\}
\end{equation}
be the solution space of all $\mathbf{e} \in \mathbb{F}_2^n$ fulfilling
$\mathbf{H} \mathbf{e} = \mathbf{s}$.
Suppose we augment the binary linear system of equations in the following way:
\begin{enumerate}[label=(\arabic*)]
\item \ac{AVN} insertion: Append an all-zero column to the right of $\mathbf{H}$.
\item \ac{ACN} insertion: Append a new row
$\mathbf{c}_{\mathrm{aug}}\tp = (\mathbf{c'}\tp,1) \in \mathbb{F}_2^{1 \times (n+1)}$
after the last row of $\mathbf{H}$. $\mathbf{c'}$ can be chosen arbitrarily
but is fixed thereafter.\footnote{Note specifically that
$\mathbf{c}'^\top$ can be linearly independent of the rows of $\mathbf{H}$.}
\item Transform $\mathbf{s}$ to $\mathbf{s}_{\mathrm{aug}} \in \mathbb{F}_2^{m+1}$
by appending an arbitrary but fixed bit $s'$.
\end{enumerate}
The results of this process are shown in Eq.~\ref{eq:aug} and
Eq.~\ref{eq:system2}.
\begin{equation}\label{eq:aug}
\mathbf{H}_{\mathrm{aug}}=
\raisebox{0.9ex}{$m+1$}\!\Bigg\{
\overbrace{
\begin{bmatrix}
\mathbf{H} & \mathbf{0}\\
\mathbf{c'}\tp & 1
\end{bmatrix}
}^{n+1}\qquad
\mathbf{s}_{\mathrm{aug}} =
\begin{bmatrix}
\mathbf{s}\\
s'
\end{bmatrix}
\end{equation}
\begin{equation}\label{eq:system2}
 \mathcal{E}(\mathbf{H}_{\mathrm{aug}},\mathbf{s}_{\mathrm{aug}})\coloneqq\{\mathbf{e}_{\mathrm{aug}}\in\mathbb{F}_2^{n+1}:\mathbf{H}_{\mathrm{aug}}\mathbf{e}_{\mathrm{aug}}=\mathbf{s}_{\mathrm{aug}}\}
\end{equation}
Here, $\mathbf{e}_{\mathrm{aug}}=(\mathbf{e},e')$ denotes $\mathbf{e}$
extended by a new bit $e'$.
\begin{lemma}[Augmentation]\label{lem:aug}
For $\mathbf{H}_{\mathrm{aug}},\mathbf{s}_{\mathrm{aug}}$ as in
Eq.~\ref{eq:aug}, the projection
$\pi:\mathcal{E}(\mathbf{H}_{\mathrm{aug}},\mathbf{s}_{\mathrm{aug}})
\to\mathcal{E}(\mathbf{H},\mathbf{s})$,
$(\mathbf{e},e') \mapsto \mathbf{e}$, is a bijection with inverse
$\mathbf{e}\mapsto(\mathbf{e},\,s'+\mathbf{c}'^\top\mathbf{e})$.
\end{lemma}
\begin{proof}
The augmented system consists of $\mathbf{H}\mathbf{e}=\mathbf{s}$ and
$e'=s'+\mathbf{c}'^\top\mathbf{e}$. The first equation is the original
system; the second determines $e'$ uniquely from $\mathbf{e}$ and imposes no
constraint on it.
\end{proof}
\emph{Remark:} Observe that merely appending a single row
$\mathbf{c}_{\mathrm{aug}}\tp$ together with an \ac{LLR}-$0$ prior \ac{AVN}
has no influence on \ac{BP} decoding using this augmented graph. As evident
from Eq.~\ref{eq:vn_update} and Eq.~\ref{eq:cn_update}, due to this
degree-$1$ \ac{AVN}, the corresponding \ac{ACN} will always send $L = 0$ to
all connected \acp{VN}. As we will show later, only by adding this \ac{ACN}
to other rows in $\mathbf{H}_{\mathrm{aug}}$ or setting a different prior
\ac{LLR} will this influence the decoding trajectory.

Any finite sequence of augmentation steps (1)--(3) interleaved with
elementary row additions, containing $n_{\mathrm{AVN}}$ augmentations in
total, transforms $\mathbf{H}\in\mathbb{F}_2^{m\times n}$ into a matrix of
the general block structure
\begin{equation}\label{eq:block_structure}
\mathbf{H}_{\mathrm{aug}}
=\bigl(\mathbf{H}_{\mathrm{sys}}\big|\mathbf{M}\bigr)
\in\mathbb{F}_2^{m_{\mathrm{aug}}\times(n+n_{\mathrm{AVN}})},
m_{\mathrm{aug}}=m+n_{\mathrm{AVN}}.
\end{equation}
where $\mathbf{H}_{\mathrm{sys}}\in\mathbb{F}_2^{m_{\mathrm{aug}}\times n}$
acts on the original error coordinates and
$\mathbf{M}\in\mathbb{F}_2^{m_{\mathrm{aug}}\times n_{\mathrm{AVN}}}$, the
\emph{auxiliary block}, contains the \ac{AVN} columns. Note that in general
$\mathrm{row}(\mathbf{H}_{\mathrm{sys}})\neq\mathrm{row}(\mathbf{H})$, as the
appended rows $\mathbf{c}'^\top$ may be linearly independent of the rows of
$\mathbf{H}$; in particular, the augmented syndrome can no longer be obtained
from $\mathbf{s}$ directly. The following lemma summarizes the structural
properties of Eq.~\ref{eq:block_structure} that the later syndrome mapping relies
on.
\begin{lemma}\label{lem:structure}
For $\mathbf{H}_{\mathrm{aug}}=(\mathbf{H}_{\mathrm{sys}}\,|\,\mathbf{M})$ as
in Eq.~\ref{eq:block_structure}, it holds that
\begin{enumerate}[label=(\roman*)]
\item $\mathrm{rank}(\mathbf{M})=n_{\mathrm{AVN}}$, i.e., the auxiliary block
      has full column rank,
\item $\mathrm{rank}(\mathbf{H}_{\mathrm{aug}})
      =\mathrm{rank}(\mathbf{H})+n_{\mathrm{AVN}}=m_{\mathrm{aug}}$,
\item The subspace of $\mathrm{row}(\mathbf{H}_{\mathrm{aug}})$ consisting of
      all vectors with $0$ in the last $n_{\mathrm{AVN}}$ positions equals
      $\mathrm{row}(\mathbf{H})\times\{\mathbf{0}\}$.
\end{enumerate}
\end{lemma}
\begin{proof}
See Appendix~\ref{app:proof_structure}.
\end{proof}
With this and Alg. ~\ref{alg:avn}, a back transformation can be defined:
\begin{proposition}[Back Transformation]\label{prop:correctness}
Let $\mathbf{H}_{\mathrm{aug}}$ be as in Eq.~\ref{eq:block_structure} and
$\mathbf{H}_{\mathrm{map}}$ the output of Alg.~\ref{alg:avn}. Then (i)
$\mathrm{row}(\mathbf{H}_{\mathrm{map}})=\mathrm{row}(\mathbf{H})$, so
$\mathbf{H}_{\mathrm{map}}=\mathbf{A}\mathbf{H}$ for some
$\mathbf{A}\in\mathbb{F}_2^{m_{\mathrm{aug}}\times m}$ with zero rows at the
pivot positions, and (ii) for
$\mathbf{s}_{\mathrm{aug}}=\mathbf{A}\mathbf{s}$,
\[
\pi:\ \mathcal{E}(\mathbf{H}_{\mathrm{aug}},\mathbf{s}_{\mathrm{aug}})
\ \to\ \mathcal{E}(\mathbf{H},\mathbf{s}),
\qquad (\mathbf{e},\mathbf{e}')\mapsto\mathbf{e},
\]
is a bijection.
\end{proposition}
\begin{proof}
See Appendix~\ref{app:proof_back_transformation}.
\end{proof}
\begin{algorithm}[t]
\caption{Construction of $\mathbf{H}_{\mathrm{map}}$}
\label{alg:avn}
\renewcommand{\algorithmicrequire}{\textbf{Input:}}
\renewcommand{\algorithmicensure}{\textbf{Output:}}
\begin{algorithmic}[1]
\REQUIRE $\mathbf{H}_{\mathrm{aug}}=(\mathbf{H}_{\mathrm{sys}}\,|\,\mathbf{M})$
  with $\mathrm{rank}(\mathbf{M})=n_{\mathrm{AVN}}$ as Eq.~\ref{eq:block_structure}
\ENSURE $\mathbf{H}_{\mathrm{map}}\in\mathbb{F}_2^{m_{\mathrm{aug}}\times n}$
  with $\mathrm{row}(\mathbf{H}_{\mathrm{map}})=\mathrm{row}(\mathbf{H})$
\STATE By Gaussian elimination on $\mathbf{M}$, find a pivot row set
  $\mathcal{P}$ with $|\mathcal{P}|=n_{\mathrm{AVN}}$ and an invertible
  $\mathbf{E}$ adding pivot rows to non-pivot rows such that
  $(\mathbf{E}\mathbf{M})$ is zero on all rows outside $\mathcal{P}$
\STATE $\mathbf{H}_{\mathrm{map}}\gets\mathbf{E}\mathbf{H}_{\mathrm{sys}}$
  with the rows in $\mathcal{P}$ set to $\mathbf{0}$
\RETURN $\mathbf{H}_{\mathrm{map}}$
\end{algorithmic}
\end{algorithm}
The conventional syndrome-based \ac{BP} decoding framework can then be
readily used in the following way:
\begin{itemize}
    \item Given the syndrome $\mathbf{s}$, compute
    $\mathbf{s}_{\text{aug}} =\mathbf{A}\mathbf{s}$.
    \item Perform \ac{BP}-based decoding using $\mathbf{H}_{\mathrm{aug}}$
    and $\mathbf{s}_{\text{aug}}$, yielding $\hat{\mathbf{e}}_{\mathrm{aug}}$.
    \item Extract the estimate of the original error as
    $\hat{\mathbf{e}} = \hat{\mathbf{e}}_{\mathrm{aug}}[:n]$.
\end{itemize}
In the remainder of this work, we apply this framework to the circuit-level
decoding problem of Sec.~{\ref{sec:system_model_section}}, i.e., with $\mathbf{H}=\mathbf{H}_{\mathrm{DEM}}$
and the measured syndrome $\mathbf{s}=\mathbf{s}_{\mathrm{DEM}}$; the estimate
$\hat{\mathbf{e}}_{\mathrm{DEM}}$ is obtained analogously. The complete
decoding architecture for the proposed method is shown in
Fig.~\ref{fig:system_model}. In the following sections, we first show how the
removal of short cycles from a given parity check matrix can be formulated in
the proposed framework. After that, we interpret the syndrome-splitting
subcode approach in our framework. Finally, we propose an ensemble-split
method that combines the advantages of both approaches.
\section{Short-cycle Removal}
\subsection{Direct graph modification}
The \ac{GARI} method proposed in \cite{gari} introduces auxiliary nodes
to eliminate the $4$-cycles that the $Y$-component of the decoding
matrices induces in the \ac{DEM}. Although the underlying graph
transformation is stated for general decoding graphs, its application
in \cite{gari} is restricted to the $4$-cycles arising from $Y$-type
error mechanisms. Each such augmentation is an instance of
steps (1)--(3) of our framework: the vector $\mathbf{c}'$ is supported
on the columns of the targeted $Y$-induced biclique, and the
corresponding row additions follow. The resulting augmentation thus
takes the form of Eq.~\ref{eq:block_structure}, and
Proposition~\ref{prop:correctness} recovers the decoding-problem
equivalence established in \cite{gari}. Although \cite{gari} remarks
that the transformation can be applied recursively to the remaining
$4$-cycles, this option is not explored in its numerical evaluation.
In contrast, we apply the optimized \ac{AVN}/\ac{ACN}-based
$4$-cycle-removal strategy of \cite{yifei_1} to the $X$-type decoding
matrix.
The algorithm of \cite{yifei_1} greedily selects all-ones
blocks containing $\delta$ $4$-cycles each and removes them either by a plain
row addition, if this is possible without creating degree-$1$ or degree-$2$
check nodes (Cases~1 and~2 of \cite{yifei_1}), or by an \ac{AVN}/\ac{ACN}
insertion.
\begin{corollary}\label{cor:cycle_removal}
Alg.~\ref{alg:cycle_remove} performs only elementary row additions
(Cases~1 and~2, and adding the appended row to the rows in $\mathcal{S}_c$)
and augmentation steps with $\mathbf{c}'$ supported on $\mathcal{S}_v$. Its
output is therefore of the form of Eq.~\ref{eq:block_structure},
Lemma~\ref{lem:structure} applies, and the decoding pipeline via
Proposition~\ref{prop:correctness} can be used.
\end{corollary}
\begin{algorithm}[t]
\caption{Optimized $4$-cycle removal, adapted from \cite{yifei_1}}
\label{alg:cycle_remove}
\renewcommand{\algorithmicrequire}{\textbf{Input:}}
\renewcommand{\algorithmicensure}{\textbf{Output:}}
\begin{algorithmic}[1]
\REQUIRE $\mathbf{H}\in\mathbb{F}_2^{m\times n}$
\ENSURE $4$-cycle free $\mathbf{H}'$ of the form of
  Eq.~\ref{eq:block_structure}
\WHILE{$\mathbf{H}$ contains an all-ones block
  $(\mathcal{S}_c,\mathcal{S}_v)$ with
  $|\mathcal{S}_c|,|\mathcal{S}_v|\ge 2$}
\STATE Among all such blocks, select one maximizing
  $\delta=\binom{|\mathcal{S}_c|}{2}\binom{|\mathcal{S}_v|}{2}$
\IF{$\exists\, i_0\in\mathcal{S}_c:\;
  \mathcal{S}_v\subseteq\mathrm{supp}(\mathbf{h}_{i_0}\tp)$ and
  $|\mathrm{supp}(\mathbf{h}_{i_0}\tp)\setminus\mathcal{S}_v|\le 1$}
\STATE Add row $i_0$ to all other rows in $\mathcal{S}_c$
  \hfill (Cases 1 and 2)
\ELSE
\STATE Perform augmentation steps (1)--(2) with
  $\mathbf{c}'=\mathbbm{1}_{\mathcal{S}_v}$; add the new row to all rows in
  $\mathcal{S}_c$
\ENDIF
\ENDWHILE
\RETURN $\mathbf{H}'\gets\mathbf{H}$
\end{algorithmic}
\end{algorithm}

The following example aims to illustrate the workings of Alg.
\ref{alg:cycle_remove} and Alg. \ref{alg:avn}.

\emph{Example}: Consider the following $3 \times 8$ binary parity check matrix $\mathbf{H}$:
\[
\mathbf{H} =
\begin{tikzpicture}[baseline=(current bounding box.center)]
\matrix (m) [matrix of math nodes,nodes in empty cells,
    left delimiter={[},right delimiter={]}]
  {
    0 & 0 & 0 & 1 & 1 & 0 & 0 & 1 \\
    1 & 0 & 1 & 0 & 0 & 1 & 1 & 0 \\
    0 & 1 & 1 & 0 & 1 & 1 & 0 & 0 \\
  };
 
  \draw[rot,thick,rounded corners] 
    (m-2-3.north west) 
    rectangle 
    (m-3-6.south east);
\end{tikzpicture}
\]
 Observe that the matrix contains a single $4$-cycle spanned by the row-column pairs $(2,3),(3,3),(2,6),(3,6)$, for clarity also marked with a red rectangle.
 We now use Alg. \ref{alg:cycle_remove} to remove this cycle.
 The column indices participating in this cycle are $S_v = \{ 3,6 \}$, the row indices $S_c = \{2,3 \}$ and the number of $4$-cycles inside this block is $\delta = \binom{2}{2} \binom{2}{2} = 1$.
 As "Case 1" and "Case 2" of Alg. \ref{alg:cycle_remove} are not fulfilled, we proceed to remove this $4$-cycle by introducing an \ac{AVN} and \ac{ACN}.
First, we extend $\mathbf{H}$ with an all zero column to the right. Then, we append a new row $\mathbf{c}_{\mathrm{aug}}\tp$ to $\mathbf{H}$ which has a $1$ in the columns participating in the $4-$cycle and a $1$ in the newly created all-zero column. Finally, we add the newly created row to the rows participating in the $4-$cycle we want to remove; in our example these are rows $2$ and $3$. The result of this process is the $4 \times 9$ matrix $\mathbf{H}_{\mathrm{aug}}$:
\[
\begin{tikzpicture}[baseline=(m.center)]
\matrix (m) [
    matrix of math nodes,
    nodes in empty cells,
    nodes={minimum width=4mm, minimum height=4mm},
    left delimiter={[},
    right delimiter={]}
]{
  0 & 0 & 0 & 1 & 1 & 0 & 0 & 1 & |[text=mittelblau]|0 \\
  1 & 0 & 0 & 0 & 0 & 0 & 1 & 0 & |[text=mittelblau]|1 \\
  0 & 1 & 0 & 0 & 1 & 0 & 0 & 0 & |[text=mittelblau]|1 \\
  |[text=apfelgruen]|0 & |[text=apfelgruen]|0 &
  |[text=apfelgruen]|1 & |[text=apfelgruen]|0 &
  |[text=apfelgruen]|0 & |[text=apfelgruen]|1 &
  |[text=apfelgruen]|0 & |[text=apfelgruen]|0 &
  |[text=apfelgruen,text=mittelblau]|1 \\
};

\node[left=30mm of m.center] {$\mathbf{H}_{\mathrm{aug}}=$};

\node[text=apfelgruen,left=6mm of m-4-1] {\textbf{ACN}};
\node[text=mittelblau,above=4mm of m-1-9] {\textbf{AVN}};

\draw[decorate,decoration={brace,mirror,amplitude=4pt},thick]
  ([yshift=-2mm]m-4-1.south west) -- ([yshift=-2mm]m-4-8.south east)
  node[midway,below=5pt] {$\mathbf{H}_{\mathrm{sys}}$};
\draw[decorate,decoration={brace,mirror,amplitude=4pt},thick]
  ([yshift=-2mm]m-4-9.south west) -- ([yshift=-2mm]m-4-9.south east)
  node[midway,below=5pt] {$\mathbf{M}$};
  \draw[dashed,gray,thick]
  ([xshift=0.5mm]m-1-8.north east) -- ([xshift=0.5mm]m-4-8.south east);
\end{tikzpicture}
\]
Observe that indeed $\mathrm{row}(\mathbf{H}) \neq \mathrm{row}(\mathbf{H}_{\mathrm{aug}})$.
We now want to use Alg. \ref{alg:avn} to "invert" this mapping: As we only have $n_{\mathrm{AVN}} = 1$, $\mathbf{M}$ consists of single column. Say, we choose row $2$ as pivot row for the Gaussian elimination. We now add row $2$ onto rows $3$ and $4$ to make the last column all zero. Then we set row $2$ to zero and remove the last column, resulting in $\mathbf{H}_{\mathrm{map}}$:
\[
\begin{tikzpicture}[baseline=(m.center)]

\matrix (m) [
    matrix of math nodes,
    nodes in empty cells,
    nodes={minimum width=4mm, minimum height=4mm},
    left delimiter={[},
    right delimiter={]}
]{
  0 & 0 & 0 & 1 & 1 & 0 & 0 & 1 \\
  0 & 0 & 0 & 0 & 0 & 0 & 0 & 0 \\
  1 & 1 & 0 & 0 & 1 & 0 & 1 & 0 \\
  1 & 0 & 1 & 0 & 0 & 1 & 1 & 0 \\
};

\node[left=28mm of m.center] {$\mathbf{H}_{\mathrm{map}}=$};

\end{tikzpicture}
\]
$\mathbf{H} \neq  \mathbf{H}_{\mathrm{map}}$ but $\mathrm{row}(\mathbf{H}) = \mathrm{row}(\mathbf{H}_{\mathrm{map}})$. Note that $\mathbf{H} =  \mathbf{H}_{\mathrm{map}}$ if we had chosen row $4$ as the pivot row.
Finally, we determine $\mathbf{A}$ such that
$\mathbf{H}_{\mathrm{map}} = \mathbf{A}\mathbf{H}$ and
$\mathbf{s}_{\mathrm{aug}} = \mathbf{A}\mathbf{s}$:
\[
\mathbf{A} =
\begin{bmatrix}
1 & 0 & 0 \\
0 & 0 & 0 \\
0 & 1 & 1 \\
0 & 1 & 0
\end{bmatrix}
,
\qquad
\mathbf{s}_{\mathrm{aug}} = \mathbf{A}\mathbf{s}
\]
Note that, in accordance with Proposition~\ref{prop:correctness},
$\mathbf{A}$ has a zero row at the pivot position, so
$\mathbf{s}_{\mathrm{aug}}$ is computable from the measured syndrome alone.
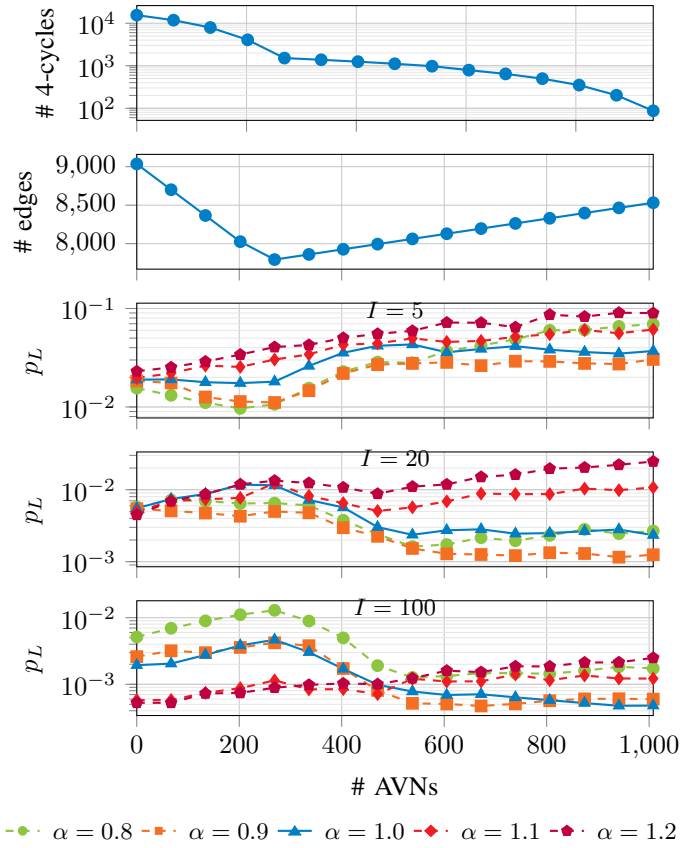
\begin{figure}[!t]
    \centering
\input{fig/fer_vs_4_cycles}
\caption{Effect of iterative $4$-cycle removal on $\mathbf{H}_{\text{DEM}}$ of the $[[72,12,6]]$ \ac{BB} code for $p=0.001$ and $r=6$ measurement rounds. The top two panels show structural quantities. The lower three panels show the logical error rate $p_L$ for different BP iteration numbers $I$ and factor node scalings $\alpha$.}
\label{fig:avn_cycle_fer}
\end{figure}

\subsection{Parity check matrix construction}
The $4$-cycle removal algorithm was recently used in the construction of parity check matrices optimized for \ac{BP} decoding of short classical channel codes in \cite{yifei_1,yifei_2}, demonstrating significant gains in performance for some codes.
For the quantum setting, \cite{KIT1} observed that adding redundant checks to the parity check matrix of \ac{CSS} codes in the code capacity model can improve performance of some codes while worsening that of others. 
In preliminary experiments on the codes considered here,
straightforward adaptations of these constructions to the \ac{DEM} parity
check matrix did not yield consistent improvements over the standard matrix
within our \ac{BP} configuration; a systematic study, including tuning of
the construction parameters, is left for future work.
\section{Affine subcode view}
Affine subcode ensemble decoding was introduced for short classical binary codes \cite{mandelbaum2026affinesubcodeensembledecoding} and recently extended to the quantum setting \cite{wursthorn2026affine}. 
The core idea is to append additional, linearly independent rows (\emph{splitters}) to the parity check matrix of a \ac{CSS} code in order to break degeneracy-induced symmetries that impair syndrome-based \ac{BP} decoding.
Since the syndrome bits of these added checks are not available before decoding, an ensemble of decoders is run, one for each hypothesized value.
This can be viewed as the check-splitting postprocessing of \cite{yin2024symbreak} moved \emph{before} \ac{BP}, with the added benefit that it can also improve the structure of the decoding graph itself.
This construction emerges naturally from our framework.
By Lemma~\ref{lem:aug}, fixing the \ac{AVN} to a hypothesized value $e'=b'$ restricts the original solution space $\mathcal{E}(\mathbf{H},\mathbf{s})$ to the affine constraint $\mathbf{c}'^\top\mathbf{e}=s'+b'$. 
The two hypotheses $b'\in\{0,1\}$ partition the solution space, so exactly one branch of the resulting two-decoder ensemble contains the true error.
Since elementary row additions leave the solution space unchanged, the same holds for the general construction of Eq.~\ref{eq:block_structure}. 
The following lemma shows which choices of $\mathbf{c}'$ provide new information for decoding:
\begin{lemma}[DEM splitter]\label{lem:split}
Let $\mathbf{e}$ denote the true physical error and let
$\mathbb{S}:=\ker\mathbf{H}_{\mathrm{DEM}}\cap\ker\mathbf{L}$.
We call the
cosets $\mathbf{e}+\mathbb{S}$ \emph{degeneracy sets}, in analogy
to~\cite{wursthorn2026affine}: adding an element of $\mathbb{S}$ to
$\mathbf{e}$ does not change the syndrome or the logical observables.
Fixing the \ac{AVN} to its true value $e'_{\mathrm{true}} = s'+\mathbf{c}'^\top\mathbf{e}$ imposes the affine constraint
$\mathbf{c}'^\top\hat{\mathbf{e}}=s'+e'_{\mathrm{true}} = \mathbf{c}'^\top\mathbf{e}$ on the candidate errors
$\hat{\mathbf{e}}\in\mathcal{E}(\mathbf{H}_{\mathrm{DEM}},\mathbf{s}_{\mathrm{DEM}})$.
This constraint
(i) provides information not contained in the measured syndrome iff
$\mathbf{c}'^\top\notin\mathrm{Row}(\mathbf{H}_{\mathrm{DEM}})$, and
(ii) splits every degeneracy set into two disjoint subsets iff
$\mathbf{c}'^\top\notin\mathrm{Row}(\mathbf{H}_{\mathrm{DEM}})+\mathrm{Row}(\mathbf{L})$.
\end{lemma}
\begin{proof}
See Appendix~\ref{app:proof_degeneracy}.
\end{proof}

As in \cite{wursthorn2026affine}, running an ensemble of two \ac{BP} decoders with the \ac{AVN} fixed to $0$ and $1$, respectively, recovers the subcode ensemble decoder.

\section{Augmented graph based ensemble decoder}
Each step $\mu$ ($0\leq\mu<P$) of the $4$-cycle removal process can be seen
as appending a new row $\mathbf{c}_{\mathrm{aug},\mu}^\top$ to the matrix
and then adding it onto existing rows.
By Lemma~\ref{lem:split}, the
appended row is a candidate splitter for subcode ensemble decoding
\cite{mandelbaum2026affinesubcodeensembledecoding} whenever its
$\mathbf{c}'^\top$ part lies outside
$\mathrm{Row}(\mathbf{H}_{\mathrm{DEM}})+\mathrm{Row}(\mathbf{L})$.
As simulations suggest, nearly all of these rows fulfill this criterion.
We therefore propose to reuse the rows generated by $4$-cycle removal as
splitter candidates for a subcode ensemble decoder.
Since this yields far more candidates than can be used ($P$), a selection rule is required; following
the observation in \cite{wursthorn2026affine}, we further restrict each
ensemble member to two decoding paths ($\Delta=2$).
We propose two greedy selection strategies:
\begin{itemize}
  \item \emph{Static}: The ensemble is built iteratively. The first member
  is the candidate with the lowest empirical logical error rate $p_L$.
  Each subsequent member is the candidate that corrects the largest number
  of sampled error events on which the current ensemble fails (the
  maximum-coverage heuristic of
  \cite{mandelbaum2025subcodeensembledecodinglinear}).
     \item \emph{Adaptive}: Standard \ac{BP} is run for $T$ iterations on the
  $4$-cycle-free matrix. If the decoder did not converge to a syndrome-valid solution, the \acp{AVN} are ranked by their \ac{LLR} magnitudes, and the ensemble is formed from the \acp{ACN}
  associated with the $L$ least reliable \acp{AVN}.
\end{itemize}
\label{sec:construction}
If multiple ensemble members produce syndrome-valid estimates, the minimum-Hamming-weight estimate is selected as the final correction.
In the following, $E$ denotes the ensemble size, i.e. the number of two decoding path ensemble members.
\section{Numerical results}
If not specified otherwise, we perform the experiments using $I=20$ iterations of syndrome-based \ac{BP} (sum-product, $\alpha=1$) decoding at an error rate $p = 0.001$ on the $[[72,12,6]]$ \ac{BB} \cite{Bravyi_2024} code. For each data point, we collect at least $200$ logical error events, defined as instances after which at least one logical observable is flipped after applying the decoded error. Dividing this by the total number of samples yields the \ac{LER}. 
Using this, we calculate the \emph{logical error rate per syndrome measurement round} \cite{Bravyi_2024} as,
\begin{equation}
    p_L = 1-(1-\ac{LER})^{\frac{1}{r}}
\end{equation}
Following common convention \cite{beni2025tesseractsearchbaseddecoderquantum}, if not mentioned otherwise, we assume $r = d$ rounds of measurements, where $d$ is the code minimum distance.
Our simulation framework is based on QUITS \cite{Kang2025quitsmodularqldpc} for the circuit-level error model and the BP-decoder implementation of \cite{roffe_decoding_2020}.
\subsection{$4-$cycle removal}
We first evaluate the effect of $4$-cycle removal on the number of edges in the decoding graph, the $4$-cycles and the error rate $p_L$ for various iteration counts $I$ and check-node scaling factors $\alpha$ during fixed steps of the removal process. 
To this end, we first run a complete $4$-cycle removal to determine the total number of removal steps. 
Then, we evenly divide this into $15$ intermediate checkpoints where we evaluate the previously mentioned quantities on the partially $4$-cycle reduced matrices. 
The results are shown in Fig. \ref{fig:avn_cycle_fer}.
As expected, the number of $4$-cycles steadily decreases during the removal process, with the first removal step removing the largest number of cycles.
This is also reflected in the edge count of the matrices: The minimum does not lie at zero $4$-cycles but arises shortly after the start of the removal process.
Note that removing single $4$-cycles with the \ac{ACN} and \ac{AVN} construction can indeed increase the number of edges in the matrix as $4$ edges are removed and $5$ edges introduced.
This reveals a tradeoff between decoding complexity (in terms of edge count) and decoding performance.
Most interesting is the $p_L$ behaviour over the removal process.
First, note that the performance of the $4$-cycle reduced augmented matrices generally improves for a  check-node scaling $\alpha=1$ and a reasonable number of BP iterations.
Remarkably, the final $4$-cycle free matrices exhibit a strong dependence on $\alpha$ in their performance.
For a small number of \ac{BP} iterations, the decoding performance of the augmented graph worsens as more \acp{AVN} are introduced. One possible explanation is that the zero-prior-LLR \acp{AVN} provide no intrinsic information while simultaneously enlarging the decoding graph.
Consequently, additional \ac{BP} iterations may be required for useful information to propagate through the augmented graph.
Furthermore, for a large iteration count, the original $\mathbf{H}_{\text{DEM}}$ without any removed $4$-cycles shows a performance nearly as good as the final $4$-cycle removed one.\footnote{We observed the same behaviour exaggerated in the code capacity model, where noise can only occur on the data qubits.}
\subsection{Ensemble Construction}
Next, we investigate the performance of the proposed ensemble based decoding method.
To this end, we first perform an ablation study where we construct static and adaptive ensembles for different parameters of $E$ and $T$.
The results are shown in Fig. \ref{fig:adaptive_fixed}.
\begin{figure}[!h]
    \centering
    \input{fig/fer_ensemble.tex}
    \caption{Error rate $p_L$ comparison between the adaptive and static ensemble decoders for $p = 0.001$, $r=6$, and $I=20$. The shaded surface represents the adaptive decoder performance for various combinations of $E$ and $T$. The red curve shows the performance of the static ensemble decoder over varying $E$.}
\label{fig:adaptive_fixed}
\end{figure}
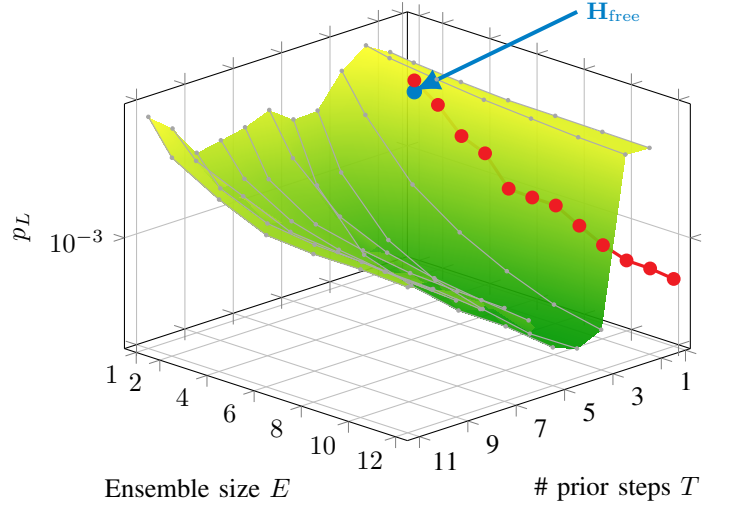
As expected, both ensemble approaches achieve lower $p_L$ as the ensemble size increases. 
Surprisingly, the static ensemble outperforms the adaptive variant over a broad range of parameter settings and performs comparably for the remaining settings. 
Note that the current $\ac{LLR}$-magnitude-based \ac{AVN} ranking might not always identify the most effective splitter candidates. 
A more refined \ac{AVN} selection strategy may therefore further improve the performance of the adaptive decoder.
In a next step, we evaluate the proposed ensemble decoders against plain \ac{BP} and the \ac{BP}+\ac{OSD} \cite{Panteleev_2021} decoder. 
The results for two different \ac{BB} codes are shown in Fig. ~\ref{fig:avn_final_ensemble_sweep}.
\begin{figure}[t]
    \centering
    \subfloat[$\lbrack\lbrack 72,12,6\rbrack\rbrack$ BB code \label{fig:avn_sweep_code_a}]{%
        \input{fig/fer_vs_p_72_12}
    }

    \vspace{0.5em}

    \subfloat[$\lbrack\lbrack 90,8,10\rbrack\rbrack$ BB code \label{fig:avn_sweep_code_b}]{%
        \input{fig/fer_vs_p_90_8}
    }

    \caption{Logical error rate per round $p_L$ of the proposed ensembles together with different baselines over the physical error rate $p$. $\mathrm{BP}_{\text{adaptive}}$ denotes the adaptive ensemble and $\mathrm{BP}_{\text{static}}$ the static variant.
    All tested decoders use a total of $I=20$ \ac{BP} iterations.}
    \label{fig:avn_final_ensemble_sweep}
\end{figure}
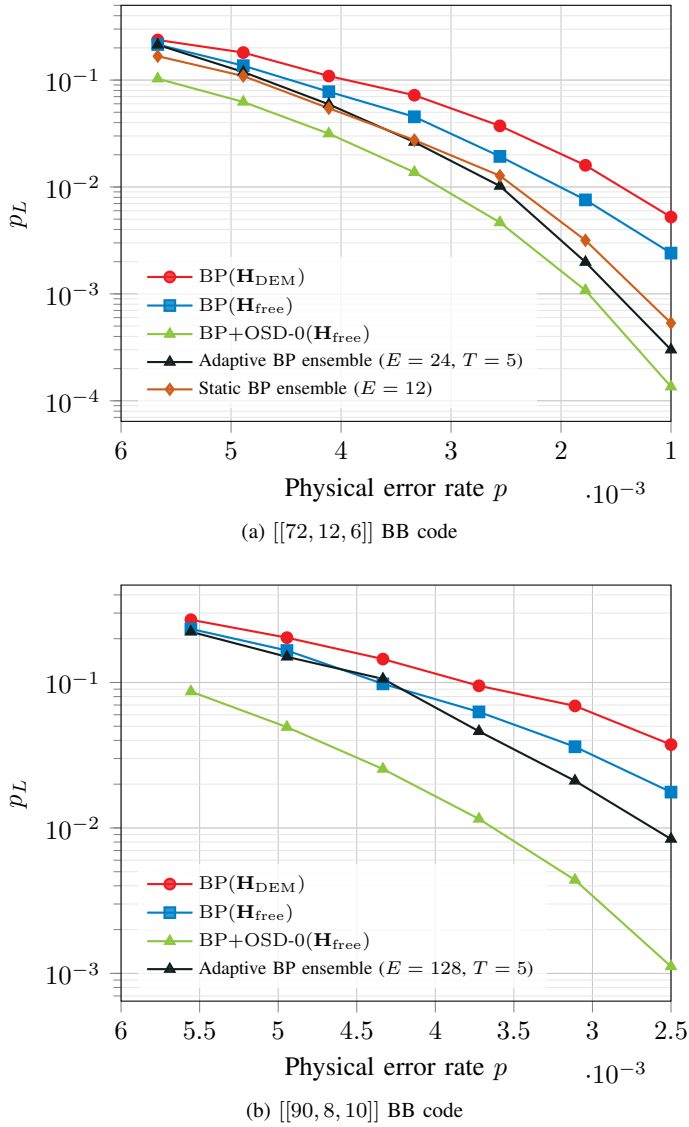
Again, removing the $4$-cycles increases the \ac{BP} decoding performance for both codes.
For the smaller $[[72,12,6]]$ code, the adaptive ensemble with $24$ ensemble members nearly reaches the performance of the BP+OSD-$0$ baseline while permitting latency comparable to plain \ac{BP} through parallel execution.
For the larger code, however, a substantial performance gap remains even with $128$ ensemble members.
This suggests that further gains may require complementary techniques, such as layered or windowed decoding, which could reduce the effective decoding problem size and improve information propagation.
\section{Conclusion}
In this work, we introduced a general framework that incorporates \acp{AVN} and \acp{ACN} into the decoding graph for binary syndrome-based decoding of \ac{QLDPC} codes while preserving the original decoding problem. The framework provides a common algebraic interpretation of graph augmentation, $4$-cycle removal and affine subcode constraints. We showed that while $4$-cycle removal can improve the performance, its benefits strongly depend on the \ac{BP} iteration count and check-node message scaling.
By reusing the auxiliary nodes generated during $4$-cycle removal as candidate degeneracy splitters, the proposed static and adaptive subcode ensembles substantially improve plain BP decoding.
Improved decoding-matrix construction methods utilizing auxiliary nodes might further improve performance.
Combining these methods with windowed decoding could be particularly effective for larger codes, as it could reduce the effective decoding problem size while retaining the benefits of graph-derived subcode ensembles.
\label{sec:conclusion}

\bibliographystyle{IEEEtran}
\bibliography{IEEEabrv,references}
\appendices

\section{Proof of Lemma~\ref{lem:structure}}
\label{app:proof_structure}

\begin{proof}
By induction over the sequence of operations. (i) and (ii): A row addition is a left-multiplication by an invertible matrix and preserves the rank of any column subset, in particular of $\mathbf{M}$ and of $\mathbf{H}_{\mathrm{aug}}$. An augmentation grows $\mathbf{M}$ to the block-triangular matrix $\bigl[\begin{smallmatrix}\mathbf{M}&\mathbf{0}\\ \mathbf{v}^\top&1\end{smallmatrix}\bigr]$; since the new column is supported only on the new row, both ranks grow by exactly one. (iii): Row additions leave the rowspace unchanged. After augmentation, the appended row is the only row with a nonzero entry in the new column; hence any element of the row space equals zero in that column has zero coefficient on the appended row and thus lies in the previous row space extended by $0$. Starting from $\mathbf{H}$ with no auxiliary columns, the claims follow.
\end{proof}
\section{Proof of Proposition~\ref{prop:correctness}}
\label{app:proof_back_transformation}
\begin{proof}
(i) Performing Gaussian elimination on $\mathbf{M}$ leaves $m_{\mathrm{aug}}-n_{\mathrm{AVN}}=\mathrm{rank}(\mathbf{H})$ rows in $\mathbf{H}_{\mathrm{aug}}$ with 
$0$ in the last $n_{\mathrm{AVN}}$ positions, which by Lemma \ref{lem:structure} (iii) span a subspace of $\mathrm{row}(\mathbf{H})\times\{\mathbf{0}\}$. So, $\mathrm{row}(\mathbf{H}_{\mathrm{map}}) \subseteq \mathrm{row}(\mathbf{H})$. As the non-zero part of $(\mathbf{E}\mathbf{M})$ has rank $n_{\mathrm{AVN}}$, the remaining rows with $0$ in the last $n_{\mathrm{AVN}}$ columns have to have rank $\mathrm{rank}(\mathbf{H})$, thus $\mathrm{row}(\mathbf{H}_{\mathrm{map}}) = \mathrm{row}(\mathbf{H})$.
(ii) Since $\mathbf{s}_{\mathrm{aug}}$ equals to $0$ on the pivot positions,
$\mathbf{E}\mathbf{s}_{\mathrm{aug}}=\mathbf{s}_{\mathrm{aug}}$ (as $\mathbf{E}$ corresponds to adding pivot rows to other rows), so the
augmented system is equivalent to the eliminated one. There, the non-pivot
equations reproduce $\mathbf{H}\hat{\mathbf{e}}=\mathbf{s}$, while the pivot
equations carry an invertible block on the auxiliary columns and, as in
Lemma~\ref{lem:aug}, merely determine $\hat{\mathbf{e}}'$ from
$\hat{\mathbf{e}}$.
\end{proof}
\section{Proof of Lemma~\ref{lem:split}}
\label{app:proof_degeneracy}
\begin{proof}
Both statements follow from the same observation. Consider a coset
$\mathbf{e}_0+\mathbb{V}$ for some subspace $\mathbb{V}$. The value
$\mathbf{c}'^\top\hat{\mathbf{e}}$ is the same for all
$\hat{\mathbf{e}}=\mathbf{e}_0+\boldsymbol{\delta}$ in this coset iff
$\mathbf{c}'^\top\boldsymbol{\delta}=0$ for all
$\boldsymbol{\delta}\in\mathbb{V}$, i.e., iff
$\mathbf{c}'^\top\in\mathbb{V}^\perp$; in this case the constraint cannot
distinguish elements of the coset. Otherwise, there is a
$\boldsymbol{\sigma}_0\in\mathbb{V}$ with
$\mathbf{c}'^\top\boldsymbol{\sigma}_0=1$. Adding $\boldsymbol{\sigma}_0$
to any element stays within the coset but flips the value of
$\mathbf{c}'^\top\hat{\mathbf{e}}$, so both values occur and the constraint
splits the coset into two nonempty parts.

Statement (i) is this observation applied to the solution space
$\mathcal{E}(\mathbf{H}_{\mathrm{DEM}},\mathbf{s}_{\mathrm{DEM}})
=\mathbf{e}+\ker\mathbf{H}_{\mathrm{DEM}}$, using
$(\ker\mathbf{H}_{\mathrm{DEM}})^\perp=\mathrm{Row}(\mathbf{H}_{\mathrm{DEM}})$.
Statement (ii) is the same applied to the degeneracy sets
$\mathbf{e}_0+\mathbb{S}$, where duality gives
$\mathbb{S}^\perp
=(\ker\mathbf{H}_{\mathrm{DEM}}\cap\ker\mathbf{L})^\perp
=\mathrm{Row}(\mathbf{H}_{\mathrm{DEM}})+\mathrm{Row}(\mathbf{L})$.
\end{proof}
\end{document}

%% file: macros.tex
\newcommand{\tp}{^{\mathsf{T}}}



%% file: acronyms.tex
\begin{acronym}
    \acro{ECC}{error-correcting code}
    \acro{MLD}{maximum likelihood decoding}
    \acro{HDD}{hard decision decoding}
    \acro{ML}{maximum likelihood}
    \acro{BP}{belief propagation}
    \acro{BER}{bit error rate}
    \acro{BPSK}{binary phase shift keying}
    \acro{QPSK}{quadrature phase shift keying}
    \acro{AWGN}{additive white Gaussian noise}
    \acro{MSE}{mean squared error}
    \acro{LLR}{log-likelihood ratio}
    \acro{MAP}{maximum a posteriori}
    \acro{BLER}{block error rate}
    \acro{NN}{neural network}
    \acro{DL}{deep learning}
    \acro{BMI}{bit-wise mutual information}
    \acro{LDPC}{low-density parity-check}
    \acro{FEC}{forward error correction}
    \acro{CRC}{cyclic redundancy check}
    \acro{QEC}{quantum error correction}
    \acro{AVN}{auxiliary variable node}
    \acro{ACN}{auxiliary check node}
    \acro{QSC}{quantum stabilizer code}
    \acro{QLDPC}{quantum low-density parity-check}
    \acro{CSS}{Calderbank--Shor--Steane}
    \acro{VN}{variable node}
    \acro{CN}{check node}
    \acro{BP2}{binary belief propagation}
    \acro{BP4}{quaternary belief propagation}
    \acro{FER}{frame error rate}
    \acro{TS}{trapping set}
    \acro{DEM}{detector error model}
    \acro{BB}{bivariate bicycle}
    \acro{LER}{logical error rate}
    \acro{ASCED}{affine subcode ensemble decoding}
    \acro{OSD}{ordered statistics decoding}
    \acro{GARI}{graph augmentation and rewiring for interference}
\end{acronym}

%% file: corporateColours.tex
\definecolor{mittelblau}{RGB}{0, 126, 198}
\definecolor{violettblau}{cmyk}{0.9, 0.6, 0, 0}
\definecolor{rot}{RGB}{238, 28 35}
\definecolor{apfelgruen}{RGB}{140, 198, 62}
\definecolor{gelb}{RGB}{1, 221, 0}
\definecolor{orange}{RGB}{244, 111, 33}
\definecolor{pink}{RGB}{237, 0, 140}
\definecolor{lila}{RGB}{128, 10, 145}
\definecolor{hellgrau}{RGB}{224, 224, 224}
\definecolor{mittelgrau}{RGB}{128, 128, 128}
\definecolor{dunkelgrau}{RGB}{80,80,80}
\definecolor{anthrazit}{RGB}{19, 31, 31}

%% file: plotcyclelists.tex
\pgfplotscreateplotcyclelist{corporate colours markers}{%
rot, every mark/.append style={fill=.!80!rot},mark=*\\%
mittelblau, every mark/.append style={fill=.!80!mittelblau},mark=square*\\%
apfelgruen, every mark/.append style={fill=.!80!apfelgruen},mark=triangle*\\%
orange, mark=star\\%
pink, every mark/.append style={fill=.!80!pink},mark=diamond*\\%
violettblau, every mark/.append style={fill=.!80!violettblau},mark=otimes*\\%
lila, mark=|\\%
gelb, every mark/.append style={fill=.!80!gelb},mark=pentagon*\\%
hellgrau, mark=text,text mark=p\\%
anthrazit, mark=text,text mark=a\\%
}

\pgfplotscreateplotcyclelist{corporate colours markers double}{%
rot, every mark/.append style={fill=.!80!rot},mark=*\\%
rot, dashed, every mark/.append style={fill=.!80!rot,solid},mark=*\\%
mittelblau, every mark/.append style={fill=.!80!mittelblau},mark=square*\\%
mittelblau, dashed, every mark/.append style={fill=.!80!mittelblau,solid},mark=square*\\%
apfelgruen, every mark/.append style={fill=.!80!apfelgruen},mark=triangle*\\%
apfelgruen, dashed, every mark/.append style={fill=.!80!apfelgruen,solid},mark=triangle*\\%
orange, mark=star\\%
orange, dashed, mark=star\\%
pink, every mark/.append style={fill=.!80!pink},mark=diamond*\\%
pink, dashed, every mark/.append style={fill=.!80!pink,solid},mark=diamond*\\%
violettblau, every mark/.append style={fill=.!80!violettblau},mark=otimes*\\%
violettblau, dashed, every mark/.append style={fill=.!80!violettblau,solid},mark=otimes*\\%
lila, mark=|\\%
lila, dashed, mark=|\\%
gelb, every mark/.append style={fill=.!80!gelb},mark=pentagon*\\%
gelb, dashed, every mark/.append style={fill=.!80!gelb,solid},mark=pentagon*\\%
}

\pgfplotsset{
  colormap/magma/.style={%
    /pgfplots/colormap={magma}{%
      rgb=(0.001462, 0.000466, 0.013866)
      rgb=(0.035520, 0.028397, 0.125209)
      rgb=(0.102815, 0.063010, 0.257854)
      rgb=(0.191460, 0.064818, 0.396152)
      rgb=(0.291366, 0.064553, 0.475462)
      rgb=(0.384299, 0.097855, 0.501002)
      rgb=(0.475780, 0.134577, 0.507921)
      rgb=(0.569172, 0.167454, 0.504105)
      rgb=(0.664915, 0.198075, 0.488836)
      rgb=(0.761077, 0.231214, 0.460162)
      rgb=(0.852126, 0.276106, 0.418573)
      rgb=(0.925937, 0.346844, 0.374959)
      rgb=(0.969680, 0.446936, 0.360311)
      rgb=(0.989363, 0.557873, 0.391671)
      rgb=(0.996580, 0.668256, 0.456192)
      rgb=(0.996727, 0.776795, 0.541039)
      rgb=(0.992440, 0.884330, 0.640099)
      rgb=(0.987053, 0.991438, 0.749504)
    },
  },
  colormap/inferno/.style={%
    /pgfplots/colormap={inferno}{%
      rgb=(0.001462, 0.000466, 0.013866)
      rgb=(0.037668, 0.025921, 0.132232)
      rgb=(0.116656, 0.047574, 0.272321)
      rgb=(0.217949, 0.036615, 0.383522)
      rgb=(0.316282, 0.053490, 0.425116)
      rgb=(0.410113, 0.087896, 0.433098)
      rgb=(0.503493, 0.121575, 0.423356)
      rgb=(0.596940, 0.154848, 0.398125)
      rgb=(0.688653, 0.192239, 0.357603)
      rgb=(0.775059, 0.239667, 0.303526)
      rgb=(0.851384, 0.302260, 0.239636)
      rgb=(0.912966, 0.381636, 0.169755)
      rgb=(0.956852, 0.475356, 0.094695)
      rgb=(0.981895, 0.579392, 0.026250)
      rgb=(0.987464, 0.690366, 0.079990)
      rgb=(0.973088, 0.805409, 0.216877)
      rgb=(0.947594, 0.917399, 0.410665)
      rgb=(0.988362, 0.998364, 0.644924)
    },
  },
  colormap/plasma/.style={%
    /pgfplots/colormap={plasma}{%
      rgb=(0.050383, 0.029803, 0.527975)
      rgb=(0.186213, 0.018803, 0.587228)
      rgb=(0.287076, 0.010855, 0.627295)
      rgb=(0.381047, 0.001814, 0.653068)
      rgb=(0.471457, 0.005678, 0.659897)
      rgb=(0.557243, 0.047331, 0.643443)
      rgb=(0.636008, 0.112092, 0.605205)
      rgb=(0.706178, 0.178437, 0.553657)
      rgb=(0.768090, 0.244817, 0.498465)
      rgb=(0.823132, 0.311261, 0.444806)
      rgb=(0.872303, 0.378774, 0.393355)
      rgb=(0.915471, 0.448807, 0.342890)
      rgb=(0.951344, 0.522850, 0.292275)
      rgb=(0.977856, 0.602051, 0.241387)
      rgb=(0.992541, 0.687030, 0.192170)
      rgb=(0.992505, 0.777967, 0.152855)
      rgb=(0.974443, 0.874622, 0.144061)
      rgb=(0.940015, 0.975158, 0.131326)
    },
  },
}

%% file: system.tikzstyles
\tikzstyle{test 1}=[fill=white, draw=black, shape=circle]
\tikzstyle{plus}=[fill=white, draw=black, shape=circle, scale=0.5]

\tikzstyle{arrow 1}=[->, -{triangle 45[length=3mm]}, draw={rgb,255: red,62; green,48; blue,255}, fill={rgb,255: red,62; green,48; blue,255}, ultra thick]
\tikzstyle{box bs}=[-, draw=mittelblau, fill = mittelblau, fill opacity=0.1, ultra thick]
\tikzstyle{box ue}=[-, draw=orange, fill = orange, fill opacity=0.1, ultra thick]
\tikzstyle{box controller}=[-, fill=rot, draw=rot, fill opacity=0.2, ultra thick, rounded corners]
\tikzstyle{box parameter}=[-, fill=apfelgruen!75, draw=apfelgruen, thick]
\tikzstyle{arrow 2}=[->, -{triangle 45[length=3mm]}, semithick]
\tikzstyle{box3}=[-, fill=white]
\tikzstyle{arrow 3}=[->, -{triangle 45[length=3mm]}, semithick, double]
\tikzstyle{arrow 4}=[->, -{triangle 45[length=3mm]}, thick]
\tikzstyle{arrow 5}=[-, thick]
\tikzstyle{arrow 6}=[->, -{triangle 45[length=3mm]}, thick, double]
\tikzstyle{arrow 7}=[-, thick, double]
\tikzstyle{arrow 8}=[-, semithick, double]
\tikzstyle{arrow 9}=[->, -{triangle 45[length=3mm]}, draw={rot}, fill={rot}, ultra thick]
\tikzstyle{arrow 10}=[-, draw={rot}, fill={rot}, ultra thick]

\tikzstyle{box 1}=[-, draw={mittelblau}, fill={mittelblau}, fill opacity=0.1, ultra thick]
\tikzstyle{box 2}=[-, draw={orange}, fill opacity=0.2, ultra thick]
\tikzstyle{box 3}=[-, fill={rot}, draw=rot, fill opacity=0.2, ultra thick, rounded corners]

%% file: fig/system_model_updated.tex
\providecolor{anthrazit}{HTML}{2F3437}
\providecolor{methodblue}{HTML}{2B6CB0}
\providecolor{apfelgruen}{HTML}{8ABD24}

\begin{tikzpicture}[
    font=\footnotesize,
    line cap=round,
    line join=round,
    wire/.style={draw=anthrazit!65,line width=0.42pt},
    cnot/.style={draw=anthrazit,line width=0.72pt},
    bpblock/.style={
        draw=methodblue!85!black,fill=methodblue!9,rounded corners=1.2pt,
        line width=0.8pt,minimum width=3.20cm,minimum height=1.48cm,
        inner sep=3pt,align=center
    },
    mapblock/.style={
        draw=apfelgruen!85!black,fill=apfelgruen!12,rounded corners=1.2pt,
        line width=0.8pt,minimum width=2.55cm,minimum height=0.50cm,
        inner xsep=2.5pt,inner ysep=1.0pt,align=center,font=\scriptsize
    },
    arr/.style={-{Triangle[length=2.0mm,width=1.48mm]},draw=anthrazit,line width=0.8pt},
    control/.style={circle,fill=anthrazit,inner sep=0pt,minimum size=3.0pt},
    target/.style={circle,draw=anthrazit,line width=0.72pt,minimum size=3.8mm,inner sep=0pt,
        path picture={
            \draw[anthrazit,line width=0.72pt](path picture bounding box.north)--(path picture bounding box.south);
            \draw[anthrazit,line width=0.72pt](path picture bounding box.west)--(path picture bounding box.east);
        }},
    meas/.style={draw=anthrazit,fill=white,line width=0.65pt,rounded corners=0.7pt,
        minimum width=4.8mm,minimum height=4.0mm,inner sep=0.4pt,font=\scriptsize},
    annot/.style={font=\scriptsize,text=anthrazit},
    processgroup/.style={draw=anthrazit!45,rounded corners=2pt,line width=0.5pt,inner sep=5pt},
    tedge/.style={draw=anthrazit!60,line width=0.38pt},
    synedge/.style={draw=methodblue!80!black,line width=0.48pt,dashed},
    tvn/.style={circle,draw=anthrazit!75,fill=white,line width=0.45pt,minimum size=2.55mm,inner sep=0pt},
    tcn/.style={rectangle,draw=anthrazit!75,fill=white,line width=0.45pt,minimum size=2.4mm,inner sep=0pt},
    auxvn/.style={tvn,draw=methodblue!90!black,fill=methodblue!28,line width=0.65pt},
    auxcn/.style={tcn,draw=methodblue!90!black,fill=methodblue!28,line width=0.65pt}
]

\foreach \y in {0,-0.40,-0.80} {\draw[wire] (0,\y) -- (2.35,\y);}
\draw[wire] (0,-1.98) -- (2.35,-1.98);
\foreach \x/\ydata in {0.48/0,1.05/-0.40,1.62/-0.80} {
    \draw[cnot] (\x,\ydata) -- (\x,-1.98);
    \node[control] at (\x,\ydata) {};
    \node[target] at (\x,-1.98) {};
}
\node[meas] at (2.35,-1.98) {$M$};
\coordinate (extract-nw) at (-0.08,0.10);
\coordinate (extract-se) at (2.62,-2.20);
\node[processgroup,fit=(extract-nw)(extract-se),label={[annot]above:\textbf{Syndrome extraction}}] (extraction) {};

\node[bpblock] (bp) at (5.82,-0.58) {};
\node[font=\footnotesize\bfseries,text=anthrazit] at ($(bp.north)+(0,-0.21)$) {BP decoder};
\draw[methodblue!35,line width=0.4pt]
    ($(bp.north west)+(0.13,-0.34)$) -- ($(bp.north east)+(-0.13,-0.34)$);

\coordinate (v1) at ($(bp.center)+(-0.73,-0.01)$);
\coordinate (v2) at ($(bp.center)+(-0.31,-0.01)$);
\coordinate (v3) at ($(bp.center)+(0.11,-0.01)$);
\coordinate (c1) at ($(bp.center)+(-0.55,-0.43)$);
\coordinate (c2) at ($(bp.center)+(-0.11,-0.43)$);
\node[annot,anchor=south] at ($(v1)!0.5!(v2)+(0,0.07)$)
    {$\hat{\mathbf e}_{\mathrm{DEM}}$};
\draw[tedge] (v1)--(c1) (v2)--(c1) (v2)--(c2) (v3)--(c2);
\node[tvn] at (v1) {};
\node[tvn] at (v2) {};
\node[auxvn] at (v3) {};
\node[tcn] at (c1) {};
\node[auxcn] at (c2) {};
\node[annot,anchor=west] at ($(bp.center)+(0.38,-0.18)$) {$\mathbf H_{\mathrm{aug}}$};

\coordinate (synstart) at ($(bp.south)+(0,0.06)$);
\coordinate (synbendone) at ($(c1.south)+(0,-0.18)$);
\coordinate (synbendtwo) at ($(c2.south)+(0,-0.18)$);
\coordinate (synpreone) at ($(c1.south)+(0,-0.065)$);
\coordinate (synpretwo) at ($(c2.south)+(0,-0.065)$);

\draw[synedge] (synstart) -- (synstart -| synbendone);
\draw[synedge] (synstart) -- (synstart -| synbendtwo);

\draw[methodblue!80!black,line width=0.48pt]
    ($(synstart -| synbendone)+(-0.7pt,0)$)
    -- (synstart -| synbendone)
    -- ($(synstart -| synbendone)+(0,0.7pt)$);
\draw[methodblue!80!black,line width=0.48pt]
    ($(synstart -| synbendtwo)+(-0.7pt,0)$)
    -- (synstart -| synbendtwo)
    -- ($(synstart -| synbendtwo)+(0,0.7pt)$);

\draw[synedge,dash phase=0pt]
    (synstart -| synbendone) -- (synpreone);
\draw[synedge,dash phase=0pt]
    (synstart -| synbendtwo) -- (synpretwo);
\draw[methodblue!80!black,line width=0.48pt,
      -{Triangle[length=1.15mm,width=0.85mm]},shorten >=0.7pt]
    (synpreone) -- (c1.south);
\draw[methodblue!80!black,line width=0.48pt,
      -{Triangle[length=1.15mm,width=0.85mm]},shorten >=0.7pt]
    (synpretwo) -- (c2.south);

\node[mapblock] (map) at (5.82,-1.98)
{$\mathbf{s}_{\mathrm{aug}}=\mathbf{A}\mathbf{s}_{\mathrm{DEM}}$};

\node[processgroup,fit=(bp)(map),label={[annot]above:\textbf{Classical processing}}] (classical) {};

\draw[arr] (2.65,-1.98)--node[above,annot] {$\mathbf{s}_{\mathrm{DEM}}$} (map.west);
\draw[arr] (map.north)--node[right,annot] {$\mathbf s_{\mathrm{aug}}$} (bp.south);
\draw[arr] (bp.east)--++(0.84,0) node[below,pos=0.72,annot] {$\hat{\mathbf e}_{\mathrm{DEM}}$};
\end{tikzpicture}

%% file: fig/fer_vs_4_cycles.tex
\begin{tikzpicture}

\begin{groupplot}[
group style={
group size=1 by 5,
vertical sep=0.45cm,
x descriptions at=edge bottom,
},
width=0.95\linewidth,
height=0.35\linewidth,
grid=both,
major grid style={line width=0.2pt, draw=gray!45},
minor grid style={line width=0.1pt, draw=gray!20},
tick align=outside,
tick pos=left,
xlabel={\# AVNs},
xmin=0,
clip=false,
enlarge x limits=false,
]

\nextgroupplot[
ylabel={\# 4-cycles},
ymode=log,
xlabel={},
xticklabels={},
]

\addplot[
thick,
mark=*,
color=mittelblau,
] table[
col sep=comma,
x=steps,
y=cycles,
restrict expr to domain={\thisrow{cycles}}{1:1e20},
] {plot_data/remove4_p0p001_r6_it5_s10.csv};

\nextgroupplot[
ylabel={\# edges},
xlabel={},
xticklabels={},
]

\addplot[
thick,
mark=*,
color=mittelblau,
] table[
col sep=comma,
x=steps,
y=edges,
] {plot_data/remove4_p0p001_r6_it5_s10.csv};

\nextgroupplot[
ylabel={$p_L$},
title={$I=5$},
title style={
at={(0.5,0.96)},
anchor=north,
font=\small,
},
ymode=log,
xlabel={},
xticklabels={},
]

\addplot[
thick,
dashed,
mark=*,
mark options={solid},
color=apfelgruen,
] table[
col sep=comma,
x=steps,
y=lfr,
] {plot_data/remove4_p0p001_r6_it5_s08.csv};

\addplot[
thick,
dashed,
mark=square*,
mark options={solid},
color=orange,
] table[
col sep=comma,
x=steps,
y=lfr,
] {plot_data/remove4_p0p001_r6_it5_s09.csv};

\addplot[
thick,
mark=triangle*,
color=mittelblau,
] table[
col sep=comma,
x=steps,
y=lfr,
] {plot_data/remove4_p0p001_r6_it5_s10.csv};

\addplot[
thick,
dashed,
mark=diamond*,
mark options={solid},
color=rot,
] table[
col sep=comma,
x=steps,
y=lfr,
] {plot_data/remove4_p0p001_r6_it5_s11.csv};

\addplot[
thick,
dashed,
mark=pentagon*,
mark options={solid},
color=purple,
] table[
col sep=comma,
x=steps,
y=lfr,
] {plot_data/remove4_p0p001_r6_it5_s12.csv};

\nextgroupplot[
ylabel={$p_L$},
title={$I=20$},
title style={
at={(0.5,0.96)},
anchor=north,
font=\small,
},
ymode=log,
xlabel={},
xticklabels={},
]

\addplot[
thick,
dashed,
mark=*,
mark options={solid},
color=apfelgruen,
] table[
col sep=comma,
x=steps,
y=lfr,
] {plot_data/remove4_p0p001_r6_it20_s08.csv};

\addplot[
thick,
dashed,
mark=square*,
mark options={solid},
color=orange,
] table[
col sep=comma,
x=steps,
y=lfr,
] {plot_data/remove4_p0p001_r6_it20_s09.csv};

\addplot[
thick,
mark=triangle*,
color=mittelblau,
] table[
col sep=comma,
x=steps,
y=lfr,
] {plot_data/remove4_p0p001_r6_it20_s10.csv};

\addplot[
thick,
dashed,
mark=diamond*,
mark options={solid},
color=rot,
] table[
col sep=comma,
x=steps,
y=lfr,
] {plot_data/remove4_p0p001_r6_it20_s11.csv};

\addplot[
thick,
dashed,
mark=pentagon*,
mark options={solid},
color=purple,
] table[
col sep=comma,
x=steps,
y=lfr,
] {plot_data/remove4_p0p001_r6_it20_s12.csv};

\nextgroupplot[
ylabel={$p_L$},
title={$I=100$},
title style={
at={(0.5,0.96)},
anchor=north,
font=\small,
},
ymode=log,
xlabel={\# AVNs},
]

\addplot[
thick,
dashed,
mark=*,
mark options={solid},
color=apfelgruen,
] table[
col sep=comma,
x=steps,
y=lfr,
] {plot_data/remove4_p0p001_r6_it100_s08.csv};

\addplot[
thick,
dashed,
mark=square*,
mark options={solid},
color=orange,
] table[
col sep=comma,
x=steps,
y=lfr,
] {plot_data/remove4_p0p001_r6_it100_s09.csv};

\addplot[
thick,
mark=triangle*,
color=mittelblau,
] table[
col sep=comma,
x=steps,
y=lfr,
] {plot_data/remove4_p0p001_r6_it100_s10.csv};

\addplot[
thick,
dashed,
mark=diamond*,
mark options={solid},
color=rot,
] table[
col sep=comma,
x=steps,
y=lfr,
] {plot_data/remove4_p0p001_r6_it100_s11.csv};

\addplot[
thick,
dashed,
mark=pentagon*,
mark options={solid},
color=purple,
] table[
col sep=comma,
x=steps,
y=lfr,
] {plot_data/remove4_p0p001_r6_it100_s12.csv};

\end{groupplot}

\end{tikzpicture}

\vspace{0.4em}

\begin{tikzpicture}
\draw[apfelgruen, thick, dashed] (0.0,0) -- (0.45,0);
\filldraw[apfelgruen] (0.225,0) circle (1.6pt);
\node[right] at (0.50,0) {\small $\alpha=0.8$};

\draw[orange, thick, dashed] (1.8,0) -- (2.25,0);
\filldraw[orange] (1.98,-0.045) rectangle (2.07,0.045);
\node[right] at (2.30,0) {\small $\alpha=0.9$};

\draw[mittelblau, thick] (3.6,0) -- (4.05,0);
\filldraw[mittelblau] (3.825,0.065) -- (3.745,-0.055) -- (3.905,-0.055) -- cycle;
\node[right] at (4.10,0) {\small $\alpha=1.0$};

\draw[rot, thick, dashed] (5.4,0) -- (5.85,0);
\filldraw[rot] (5.625,0.075) -- (5.705,0) -- (5.625,-0.075) -- (5.545,0) -- cycle;
\node[right] at (5.90,0) {\small $\alpha=1.1$};

\draw[purple, thick, dashed] (7.2,0) -- (7.65,0);
\node[
regular polygon,
regular polygon sides=5,
fill=purple,
draw=purple,
minimum size=4.2pt,
inner sep=0pt
] at (7.425,0) {};
\node[right] at (7.70,0) {\small $\alpha=1.2$};

\end{tikzpicture}

%% file: fig/fer_ensemble.tex
\pgfplotstableread{
E T LFR
1  1  0.004729908101955171
2  1  0.004533838244279664
4  1  0.004279235756078137
6  1  0.004103162986920816
8  1  0.004083608950770978
10 1  0.004044506636413536
12 1  0.0040054119965501345

1  2  0.005843977058334726
2  2  0.005424895569168031
4  2  0.004987706919882151
6  2  0.004627274123821357
8  2  0.004399968145913702
10 2  0.004191832036325449
12 2  0.0040783948653222435

1  3  0.004522667690923909
2  3  0.002667295687064919
4  3  0.001212302932586451
6  3  0.0007564147212378858
8  3  0.0005376496099147587
10 3  0.00041353270068844594
12 3  0.00035863829662741153

1  4  0.002850423937461599
2  4  0.0012959184045384653
4  4  0.0006068706383701716
6  4  0.000423767513223261
8  4  0.00034742304921819844
10 4  0.0003038107266729595
12 4  0.00030623338392221644

1  5  0.0028094485313890916
2  5  0.0012048243811895176
4  5  0.0005488004521528866
6  5  0.000438919783110725
8  5  0.0003817455201009512
10 5  0.00037873680125266596
12 5  0.0003486521028355227

1  6  0.003612437656821288
2  6  0.0020055443669430018
4  6  0.0010716262002568033
6  6  0.0007938008379025918
8  6  0.0006837786800611978
10 6  0.0006212930626675606
12 6  0.0005788139982193519

1  7  0.003151057188248929
2  7  0.0016556601868382703
4  7  0.0010735409159771114
6  7  0.0009193122480711802
8  7  0.0008261154234791013
10 7  0.0007938651083714277
12 7  0.0007759505152781232

1  8  0.00292280745564788
2  8  0.0018650098155877926
4  8  0.0012590004501858987
6  8  0.0010398260659671976
8  8  0.0009465469804661453
10 8  0.0009181663349755453
12 8  0.000857364236321212

1  9  0.002731766855875617
2  9  0.0018202789656158735
4  9  0.0013660911316589752
6  9  0.0012539242493091463
8  9  0.0012197990010215376
10 9  0.0011418203176402075
12 9  0.0010882276239980726

1  10 0.00435777621948763
2  10 0.0030417790842250936
4  10 0.002275285984215003
6  10 0.0016939040850176834
8  10 0.0015791799119483363
10 10 0.0014847506476626249
12 10 0.0014308111088372488

1  11 0.00578642775781224
2  11 0.0035714928053628414
4  11 0.0024306241962480657
6  11 0.0018124741822073132
8  11 0.0017199165716691445
10 11 0.0017199165716691445
12 11 0.0016582353276057837
}\adaptiveData

\pgfplotstableread{
E LFR
1  0.002797266724234082
2  0.002190611197689818
3  0.001556431238996181
4  0.001354103387415484
5  0.0009053892186147161
6  0.0008908706599951577
7  0.0008861340640485826
8  0.0007388264445458148
9  0.0006228596274587828
10 0.0005572587940854135
11 0.0005544779023927138
12 0.0005323400728040051
}\staticData

\begin{tikzpicture}

\begin{axis}[
    width=0.5\textwidth,
    height=0.4\textwidth,
    view={45}{28},
    xlabel={Ensemble size $E$},
    ylabel={\# prior steps $T$},
    zlabel={$p_L$},
    xmin=0.5,
    xmax=12.5,
    ymin=-0.2,
    ymax=11.5,
    y dir=reverse,
    zmode=log,
    log basis z={10},
    zmin=2e-4,
    zmax=7e-3,
    grid=both,
    tick align=outside,
    xtick={1,2,4,6,8,10,12},
    ytick={0,1,3,5,7,9,11},
    yticklabels={,1,3,5,7,9,11},
    colormap name=mysummer,
    clip=false,
]

\addplot3[
    only marks,
    mark=*,
    mark size=2.7pt,
    color=mittelblau,
    forget plot,
] coordinates {
    (1,0,0.0023731905391790598)
};

\addplot3[
    very thick,
    mark=*,
    mark size=2.0pt,
    color=rot,
    opacity=1,
    forget plot,
] table[
    x=E,
    y expr=0,
    z=LFR,
] {\staticData};

\addplot3[
    surf,
    mesh/cols=7,
    mesh/rows=11,
    mesh/ordering=x varies,
    shader=interp,
    opacity=0.5,
    point meta=z,
    z buffer=sort,
] table[
    x=E,
    y=T,
    z=LFR,
] {\adaptiveData};

\addplot3[
    surf,
    mesh/cols=7,
    mesh/rows=11,
    mesh/ordering=x varies,
    shader=interp,
    opacity=0.85,
    point meta=z,
    z buffer=sort,
] table[
    x=E,
    y=T,
    z=LFR,
] {\adaptiveData};

\pgfplotsinvokeforeach{1,2,4,6,8,10,12}{
    \addplot3[
        gray!65,
        line width=0.45pt,
        mark=none,
        forget plot,
        unbounded coords=discard,
    ] table[
        x=E,
        y=T,
        z=LFR,
        restrict expr to domain={\thisrow{E}}{#1:#1},
    ] {\adaptiveData};
}

\pgfplotsinvokeforeach{1,2,3,4,5,6,7,8,9,10,11}{
    \addplot3[
        gray!65,
        line width=0.45pt,
        mark=none,
        forget plot,
        unbounded coords=discard,
    ] table[
        x=E,
        y=T,
        z=LFR,
        restrict expr to domain={\thisrow{T}}{#1:#1},
    ] {\adaptiveData};
}

\addplot3[
    only marks,
    mark=*,
    mark size=0.65pt,
    color=gray!70,
    forget plot,
] table[
    x=E,
    y=T,
    z=LFR,
] {\adaptiveData};

\node[
    anchor=west,
    font=\small,
    text=mittelblau,
] (hfree) at (axis description cs:0.8,1)
    {$\mathbf{H}_{\mathrm{free}}$};

\draw[
    -{Latex[length=3.4mm,width=3.0mm]},
    line width=1.5pt,
    color=mittelblau,
] (hfree.west)
  -- (axis cs:1,0,0.002407298439506156);

\end{axis}
\end{tikzpicture}

%% file: fig/fer_vs_p_72_12.tex
\begin{tikzpicture}

\begin{axis}[
   width=\linewidth,
    height=0.8\linewidth,
    grid=both,
    major grid style={line width=0.2pt, draw=gray!45},
    minor grid style={line width=0.1pt, draw=gray!20},
    tick align=outside,
    tick pos=left,
    xlabel={Physical error rate $p$},
    ylabel={$p_L$},
    x dir=reverse,
    xmin=1e-3,
    xmax=6e-3,
    restrict x to domain=1e-3:6e-3,
    unbounded coords=discard,
    ymode=log,
    legend style={
        at={(0.03,0.03)},
        anchor=south west,
        draw=none,
        fill=white,
        fill opacity=0.8,
        text opacity=1,
        font=\scriptsize,
    },
    legend cell align=left,
]

\addplot[
    mark=*,
    mark options={fill=rot, draw=rot},
    thick,
    color=rot
] table[
    x=p,
    y=LFR,
    col sep=comma,
] {plot_data/72_12_p_sweep_r6_it20_original_bp.csv};

\addplot[
    mark=square*,
    mark options={fill=mittelblau, draw=mittelblau},
    thick,
    color=mittelblau
] table[
    x=p,
    y=LFR,
    col sep=comma,
] {plot_data/72_12_p_sweep_r6_it20_remove4_bp.csv};

\addplot[
    mark=triangle*,
    mark options={fill=apfelgruen, draw=apfelgruen},
    thick,
    color=apfelgruen
] table[
    x=p,
    y=LFR,
    col sep=comma,
] {plot_data/72_12_p_sweep_r6_it20_remove4_bp_osd0.csv};
\addplot[
    mark=triangle*,
    mark options={fill=anthrazit, draw=anthrazit},
    thick,
    color=anthrazit
] table[
    x=p,
    y=LFR,
    col sep=comma,
] {plot_data/72_12_p_sweep_r6_it20_prior5_E24.csv};
\addplot[
    mark=diamond*,
    mark options={fill=orange!85!black, draw=orange!85!black},
    thick,
    color=orange!85!black
] table[
    x=p,
    y=LFR,
    col sep=comma,
] {plot_data/72_12_static_p_sweep_r6_it20_E12.csv};

\legend{
    {$\mathrm{BP}(\mathbf{H}_{\mathrm{DEM}})$},
    {$\mathrm{BP}(\mathbf{H}_{\mathrm{free}})$},
    {$\mathrm{BP{+}OSD\text{-}0}(\mathbf{H}_{\mathrm{free}})$},
    {Adaptive BP ensemble ($E=24$, $T=5$)},
    {Static BP ensemble ($E=12$)}
}
\end{axis}

\end{tikzpicture}

%% file: fig/fer_vs_p_90_8.tex
\begin{tikzpicture}

\begin{axis}[
   width=\linewidth,
    height=0.8\linewidth,
    grid=both,
    major grid style={line width=0.2pt, draw=gray!45},
    minor grid style={line width=0.1pt, draw=gray!20},
    tick align=outside,
    tick pos=left,
    xlabel={Physical error rate $p$},
    ylabel={$p_L$},
    x dir=reverse,
    xmin=2.5e-3,
    xmax=6e-3,
    restrict x to domain=2.5e-3:6e-3,
    unbounded coords=discard,
    ymode=log,
    legend style={
        at={(0.03,0.03)},
        anchor=south west,
        draw=none,
        fill=white,
        fill opacity=0.8,
        text opacity=1,
        font=\scriptsize,
    },
    legend cell align=left,
]

\addplot[
    mark=*,
    mark options={fill=rot, draw=rot},
    thick,
    color=rot
] table[
    x=p,
    y=LFR,
    col sep=comma,
] {plot_data/90_8_p_sweep_r10_it20_original_bp.csv};

\addplot[
    mark=square*,
    mark options={fill=mittelblau, draw=mittelblau},
    thick,
    color=mittelblau
] table[
    x=p,
    y=LFR,
    col sep=comma,
] {plot_data/90_8_p_sweep_r10_it20_remove4_bp.csv};

\addplot[
    mark=triangle*,
    mark options={fill=apfelgruen, draw=apfelgruen},
    thick,
    color=apfelgruen
] table[
    x=p,
    y=LFR,
    col sep=comma,
] {plot_data/90_8_p_sweep_r10_it20_remove4_bp_osd0.csv};

\addplot[
    mark=triangle*,
    mark options={fill=anthrazit, draw=anthrazit},
    thick,
    color=anthrazit
] table[
    x=p,
    y=LFR,
    col sep=comma,
] {plot_data/90_8_p_sweep_r10_it20_prior5_E128.csv};

\legend{
    {$\mathrm{BP}(\mathbf{H}_{\mathrm{DEM}})$},
    {$\mathrm{BP}(\mathbf{H}_{\mathrm{free}})$},
    {$\mathrm{BP{+}OSD\text{-}0}(\mathbf{H}_{\mathrm{free}})$},
    {Adaptive BP ensemble ($E=128$, $T=5$)}
}
\end{axis}

\end{tikzpicture}